\documentclass[11pt]{article}

\usepackage[margin=1in]{geometry}
\usepackage{amsmath,amssymb,amsthm,bm}
\usepackage{booktabs,tabularx}
\usepackage{tikz}
\usetikzlibrary{arrows.meta,positioning}
\usepackage[hidelinks]{hyperref}
\usepackage[round]{natbib}
\usepackage{microtype}

\theoremstyle{definition}
\newtheorem{definition}{Definition}
\theoremstyle{plain}
\newtheorem{theorem}{Theorem}
\newtheorem{proposition}{Proposition}

\newcommand{\E}{\mathbb E}
\newcommand{\I}{\mathrm I}
\newcommand{\KL}{D_{\mathrm{KL}}}

\title{Task-Sufficient Contraction:\\
Source Selection for Machine Information Interfaces}
\author{Joss Armstrong\thanks{Ericsson Ireland. ORCID: 0009-0009-3462-9679.}}
\date{1 October 2026}

\begin{document}
\maketitle

\begin{abstract}
A declared task can sometimes certify a reduced source before a downstream
encoder, codebook, rate, distortion target, or optimizer is chosen. This
paper studies when one such reduction preserves the complete downstream
problem family, a property termed Task-Sufficient Contraction. The reduced
source is fixed by the task before the later operating point is selected.
An exact contraction allows the later problem to be solved on that source
with the same result as if the full source had been retained.

For a machine with a fixed set of possible actions and a fixed loss, the paper
identifies a consumer-specific source by merging states only when every
available action has the same regret in both. For finite action sets, replacing
the richer source by this reduced source preserves the complete one-step
rate-regret curve, even though the reduction is fixed before the distortion
target is chosen. A second result gives an exact characterization for quadratic
loss on affine feasible-action sets: the canonical reduced source is the
projection onto the directions in which feasible actions can differ. Under a
fixed energy budget, this becomes centered load, while retaining only the
optimal water-filled action is too coarse. Earlier Information Bottleneck,
semantic rate-distortion, and goal-oriented quantization results are then used
to distinguish exact, architecture-conditioned, approximate, failed, and
corrected contractions. The framework suggests a way for heterogeneous
machines to exchange what a receiving task needs without first aligning their
full internal representations.
\end{abstract}

\section{Introduction}

A weather forecast cannot be sufficient for the weather. It can describe the
conditions likely to matter over a region and time interval, but it cannot
settle every physical consequence of what will actually occur. It cannot determine
whether one particular gust, interacting with the exact local mechanics of
one particular building, will dislodge a roof tile. Yet consider a wind turbine
controller making a high-wind cut-out decision. Real turbine systems make
high-wind operating and shutdown decisions from wind-speed thresholds, and
wind forecasts are used operationally to anticipate those cut-outs
\citep{SmithNREL2012,Chou2025}. For that declared decision, the forecast's
prediction of average wind speed over the relevant period is sufficient. The
weather remains richer than the information problem presented to the machine.

What changed is not the weather and not the forecast. What changed is the
question.

Task-aware compression and representation learning already ask which
information matters to a downstream task. The claim here is more specific.
Sometimes the task can identify a smaller source before the later information
problem is optimized. That reduction is not merely the representation chosen
for one rate, distortion target, codebook size, or other operating point. The
same reduced source remains valid across the range of operating points covered
by the result.

The central question is therefore:

\begin{quote}
Can a declared task determine a reduced source before the downstream encoder,
codebook, rate, distortion target, or optimizer is chosen, such that solving
the downstream problem after the reduction gives the same result as solving it
on the richer source?
\end{quote}

In the exact case the requirement is simple to state:
\[
    \boxed{
    \text{reduce first, then optimize}
    =
    \text{optimize on the richer source}.}
\]
Such reductions are called \emph{Task-Sufficient Contractions}.

The distinction becomes more consequential when machines increasingly consume
one another's outputs. In an agentic economy, a model, agent, router,
controller, or service may need to act on another system without access to its
internal representation. A certified task-sufficient interface can let the
receiver act without first reconstructing or aligning the producer's internal
state, provided every distinction capable of changing the receiving decision
is preserved. The value of that possibility grows with the number and
heterogeneity of machine-to-machine handoffs. Opaque-agent protocols already
make this interface boundary an explicit engineering concern
\citep{A2A2026}.

The paper makes three connected contributions. First, it develops a common
certification framework for task-determined source reductions that remain
exact across a declared downstream optimization family, with the reduction
fixed before the later operating point is chosen. Second,
for a fixed consumer, it identifies the consumer-specific regret-complete
source $G_{\mathcal C}$ and proves that, for finite action sets, replacing the
richer source by $G_{\mathcal C}$ leaves the entire one-step rate-regret curve
unchanged. The reduction is therefore fixed before the distortion target is
selected. Third, for quadratic loss on an affine feasible-action set, the paper
characterizes the exact source directly from the directions in which feasible
actions can differ. Under fixed-energy scheduling this source is centered load,
while the optimal water-filled action is too coarse.

The framework grew out of several earlier reductions. An Information Bottleneck
result, \emph{Source-Side Sufficiency for the Information Bottleneck: Exact
Reduction and Finite-Block Equivalence}, was first posted in April 2026
\citep{ArmstrongIB}. There, a sufficient statistic can replace the source
before any point on the information curve is selected. A separately developed
semantic rate-distortion result then showed an exact reduced problem created
by a designed semantic source and its induced loss \citep{ArmstrongSK}. The
companion paper \emph{Task-Relative Information Contracts: Sufficiency,
Evidence, and Hierarchical Preservation in Multi-Agent Systems} calls the
interface that preserves everything $X$ tells a receiver about a named
criterion $Y$ a \emph{task-relative information contract}, or TRI contract
\citep{ArmstrongContract2026}. It also studies how such a contract can be
generated, what level of detail finite evidence can support, and when the same
parent-facing contract survives hierarchical refinement.

The earlier Information Bottleneck, semantic rate-distortion, and
goal-oriented quantization results are revisited as evidence for the framework,
not as three new contributions. Their comparison distinguishes exact
pre-optimization contraction from architecture-conditioned reduction,
approximate reduction, candidate failure, and corrected contraction. In that
sense the retrospective sections explain why several independently developed
reductions belonged to the same reduce-first problem.

The theory is static. It does not claim that a reduction licensed for the
current problem remains valid after acting through the interface changes the
state, learner, candidate set, action menu, criterion, or dependence
structure. It also does not compare task-sufficient sources across a sequence
of changing decision problems or define the direction in which representation
resolution should change when the action-consequence structure moves. Those
recursive questions are deliberately left outside the present paper.

\section{Contraction before optimization}

Task-Sufficient Contraction is not one universal formula.
The recurring structure is:
\[
    \boxed{
    \text{richer source}
    \longrightarrow
    \text{task-certified reduced source}
    \longrightarrow
    \text{later optimization}.}
\]

The reduced source depends on the problem. In the Information Bottleneck it
is a sufficient statistic $\phi(T)$. When every later use of a criterion $Y$
must remain possible, it is the TRI contract $S_Y$. For the
consumer-specific construction below it is $G_{\mathcal C}$. In constrained
quadratic goal-oriented quantization it is the part of the load vector that
can change one feasible action relative to another.

A useful certification pattern has four parts.

\begin{enumerate}
\item The reduction is fixed from the stated task before the later operating
      point is chosen.
\item Every feasible full-source solution has a reduced counterpart that
      preserves the task-relevant quantity and is no worse in the optimized
      resource.
\item Every feasible reduced solution can be pulled back to an admissible
      full-source solution with the same task value.
\item Any additional constraints claimed by the theorem survive the reduction.
\end{enumerate}

The proof takes a different form in each problem. In the IB result, an
encoder is averaged over source states that the relevance variable cannot
distinguish, and a reduced solution can be used again on the full source. The
rate-regret theorem below does the same thing for source states that give every
action the same regret. In semantic rate-distortion, exactness depends on the
chosen architecture and loss. In constrained GOQ, the right reduction is found
by asking how feasible actions can actually differ.

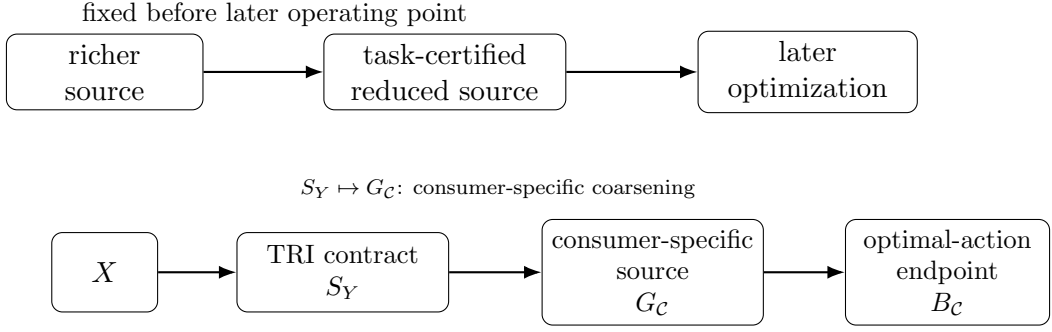
\begin{figure}[t]
\centering
\begin{tikzpicture}[>=Latex, every node/.style={align=center}]

\node[draw, rounded corners, minimum width=2.6cm, minimum height=0.95cm] (r) at (0,0)
{richer\\source};
\node[draw, rounded corners, minimum width=3.2cm, minimum height=0.95cm] (q) at (4.5,0)
{task-certified\\reduced source};
\node[draw, rounded corners, minimum width=2.9cm, minimum height=0.95cm] (o) at (9.3,0)
{later\\optimization};
\draw[->,thick] (r) -- (q);
\draw[->,thick] (q) -- (o);
\node[font=\small] at (2.25,0.76) {fixed before later operating point};

\node[draw, rounded corners, minimum width=1.4cm, minimum height=1.05cm] (x) at (0,-2.65)
{$X$};
\node[draw, rounded corners, minimum width=2.8cm, minimum height=1.05cm, font=\small] (sy) at (3.15,-2.65)
{TRI contract\\$S_Y$};
\node[draw, rounded corners, minimum width=2.8cm, minimum height=1.05cm, font=\small] (g) at (7.25,-2.65)
{consumer-specific\\source\\$G_{\mathcal C}$};
\node[draw, rounded corners, minimum width=2.7cm, minimum height=1.05cm, font=\small] (b) at (11.15,-2.65)
{optimal-action\\endpoint\\$B_{\mathcal C}$};

\draw[->,thick] (x) -- (sy);
\draw[->,thick] (sy) -- (g);
\draw[->,thick] (g) -- (b);
\node[font=\scriptsize] at (5.20,-1.55) {$S_Y \mapsto G_{\mathcal C}$: consumer-specific coarsening};

\end{tikzpicture}
\caption{Top: the general Task-Sufficient Contraction pattern. Bottom: one
important nested case,
\(X \rightarrow S_Y \rightarrow G_{\mathcal C} \rightarrow B_{\mathcal C}\).
Here \(S_Y\) is the TRI contract and preserves everything \(X\) tells the
receiver about the named criterion,
\(G_{\mathcal C}\) preserves the distinctions that can change the relative
cost of the consumer's available actions, and \(B_{\mathcal C}\) keeps only
the optimal full-information action. The lower arrows denote deterministic
factorisation or coarsening, not causal flow. This lower chain is one
application of TSC, not a required sequence for every problem.}
\label{fig:tsc}
\end{figure}

Figure~\ref{fig:tsc} also shows the relationship to the companion
TRI-contract paper. The contract $S_Y$ keeps everything $X$ tells the receiver
about $Y$. A particular consumer may need less. The two reductions can
therefore be nested without making one a required first step for the other.

\section{Consumer-specific task-sufficient sources}

Let $X$ denote the producer-side source and let $Y$ be the criterion used to
judge the receiver's output. A declared consumer specifies the actions it can
take and the loss used to evaluate them. $X$ is taken to have finite support in
the main definitions so that the reduced source is finite. The rate-regret
theorem later also assumes finitely many actions. The continuous GOQ example
does not use that coding theorem and is handled directly.

\begin{definition}[Machine consumer]
A \emph{machine consumer} is a pair
\[
    \mathcal C=(\mathcal A,\ell),
\]
where $\mathcal A$ is the set of admissible outputs or actions and
$\ell:\mathcal A\times\mathcal Y\to\mathbb R$ is the loss by which those
outputs are evaluated.
\end{definition}

The definition does not depend on how the machine is implemented. The
consumer may be a rule set, classifier, optimizer, controller, language model,
or compound agent. What matters here is only what it can output and how those
outputs are judged.

For each source state $x$, define the conditional risk
\begin{equation}
    r_x(a):=\E[\ell(a,Y)\mid X=x],
    \qquad a\in\mathcal A,
    \label{eq:risk}
\end{equation}
and the full-information Bayes envelope
\begin{equation}
    b(x):=\inf_{a\in\mathcal A} r_x(a).
    \label{eq:bayes-envelope}
\end{equation}
Assume $b(x)$ is finite for every source state and define the conditional
regret
\begin{equation}
    g_x(a):=r_x(a)-b(x)\ge0.
    \label{eq:regret-profile}
\end{equation}
The quantity $g_x(a)$ is how much worse action $a$ is than the best action at
state $x$. Because these values depend on the conditional law of $Y$ given
$X$, the reduced source depends on both the source law and the declared
consumer.

Producer-state task losses are included as a special case: for any
$L(x,a)$, take $Y=X$ and $\ell(a,Y)=L(Y,a)$. This is the form used later
in the GOQ scheduling example.

\subsection{Optimal action and the regret-complete source}

When the Bayes-optimal action is unique for every source state, write
\[
    a^*(x):=\arg\min_{a\in\mathcal A}r_x(a).
\]

\begin{definition}[Bayes quotient]
Under unique Bayes actions, define
\[
    x\sim_Bx'
    \quad\Longleftrightarrow\quad
    a^*(x)=a^*(x').
\]
A canonical representative is $B_{\mathcal C}:=a^*(X)$.
\end{definition}

This fixed-loss object is the Bayes quotient of
\citet{Sevetlidis2026} and coincides with policy equivalence in the exact
single-cycle formulation of \citet{WalshRateRegret2026}. It retains what is
needed to implement the full-information optimum.

Knowing only the best action is not always enough to design an imperfect
interface. Compression or quantization can make the receiver choose a
non-optimal action, and the cost of those mistakes can differ even when the
best action is the same.

\begin{definition}[Regret-profile equivalence and consumer-specific source]
Define
\begin{equation}
    x\sim_Gx'
    \quad\Longleftrightarrow\quad
    g_x(a)=g_{x'}(a)
    \quad\text{for every }a\in\mathcal A.
    \label{eq:regret-equivalence}
\end{equation}
Let
\[
    G_{\mathcal C}:=[X]_{\sim_G}
\]
be the resulting equivalence-class label. $G_{\mathcal C}$ is called the consumer-specific regret-complete source. The same object is the regret-profile quotient.
\end{definition}

Two source states receive the same value of $G_{\mathcal C}$ when every
available action is worse than the local optimum by the same amount in both
states. The states may still differ in their absolute baseline loss. What has
been removed is only information that cannot change the relative cost of any
action available to this consumer.

Using regret as a distortion measure is already established in sequential
decision-making and appears explicitly in Walsh's single-cycle formulation
\citep{RussoVanRoy2022,ArumugamVanRoy2021,WalshRateRegret2026}. Here the
whole set of action regrets is used for a different purpose: to decide which
source states still need to be kept separate before the later problem is
solved.

\begin{definition}[Regret-complete statistic]
A deterministic statistic $Q=q(X)$ is \emph{regret-complete} for
$(P,\mathcal C)$ if there exists a function $d_Q$ such that
\[
    g_x(a)=d_Q(q(x),a)
    \qquad\text{for all }x\in\mathcal X,\ a\in\mathcal A.
\]
\end{definition}

\begin{proposition}[Coarsest pointwise regret-complete quotient]
\label{prop:coarsest-regret}
If $Q=q(X)$ is regret-complete, then there exists a deterministic map $h$ with
\[
    G_{\mathcal C}=h(Q).
\]
Thus $G_{\mathcal C}$ is the coarsest deterministic quotient of $X$ that
preserves every action regret pointwise, up to relabelling.
\end{proposition}

This minimality is pointwise. It says that any deterministic summary that
preserves the cost of every action must also distinguish every value of
$G_{\mathcal C}$. It does not claim that $G_{\mathcal C}$ is the smallest
summary that happens to preserve one particular rate-regret curve or one
chosen distortion level.

\begin{proof}
If $q(x)=q(x')$, regret-completeness gives
$g_x(a)=d_Q(q(x),a)=d_Q(q(x'),a)=g_{x'}(a)$ for every $a$. Hence $x\sim_Gx'$,
so each fibre of $Q$ lies inside one regret-equivalence class. The class label
is therefore a deterministic function of $Q$.
\end{proof}

\begin{proposition}[Regret quotient refines the Bayes quotient]
\label{prop:regret-refines-bayes}
Under unique Bayes actions, there exists a deterministic map $k$ such that
\[
    B_{\mathcal C}=k(G_{\mathcal C}).
\]
\end{proposition}

\begin{proof}
Within a regret-equivalence class the complete regret function is identical.
Its unique zero is therefore identical, and that zero is the unique
Bayes-optimal action.
\end{proof}

In plain terms, $B_{\mathcal C}$ records what the consumer should do with
full information. $G_{\mathcal C}$ also records what each possible mistake
would cost.

Recent work on action-sufficient goal representations shows that information
sufficient to recover the optimal value need not recover the optimal action
distribution \citep{Hyeon2026}. For state $S$, goal $G$ and representation
$Z=\phi(S,G)$, their action-sufficiency condition is $\I(A;G\mid S,Z)=0$,
where $A$ follows a fixed optimal policy. That condition concerns optimal
action prediction. The present construction asks a stronger preservation
question for imperfect interfaces: which source distinctions preserve the
relative cost of every feasible action? Under unique Bayes actions,
$B_{\mathcal C}$ retains the optimal-action endpoint, whereas the
consumer-specific regret-complete source $G_{\mathcal C}$ retains the
complete feasible-action regret profile. Under the finite-source and
finite-action assumptions of Theorem~\ref{thm:contraction}, the latter
preserves the entire one-step rate-regret curve with the reduction fixed
before the distortion target is selected. Two states can therefore share the
same optimal action while remaining distinct because the same nonoptimal
action incurs different regrets.

\begin{proposition}[When the Bayes quotient is regret-complete]
\label{prop:bayes-regret-coincide}
Under unique Bayes actions, the Bayes quotient is regret-complete if and only
if there exists a function
\[
    d:\mathcal A\times\mathcal A\to\mathbb R_{\ge0}
\]
such that
\begin{equation}
    g_x(a)=d(a^*(x),a)
    \qquad\text{for every }x\text{ and }a.
    \label{eq:action-factorization}
\end{equation}
\end{proposition}

\begin{proof}
If \eqref{eq:action-factorization} holds, every regret profile is determined by
$a^*(x)$. Conversely, if the regret profile is determined by the Bayes
quotient, define $d(a^*,a)$ to be the common regret value on that Bayes class.
\end{proof}

Equation~\eqref{eq:action-factorization} is exactly the condition under
which sending only the optimal action also preserves the cost of every possible
mistake.

\subsection{Simple examples}

\paragraph{Squared error.}
Let $Y\in\mathbb R$ and $\ell(a,Y)=(Y-a)^2$. Writing
$m(x)=\E[Y\mid X=x]$,
\[
    r_x(a)=\operatorname{Var}(Y\mid x)+(a-m(x))^2,
\]
so $g_x(a)=(a-m(x))^2$. Bayes action and regret profile are both generated by
$m(X)$.

\paragraph{Log loss.}
For finite $Y$, let the action be a predictive distribution $q$ and
$\ell(q,y)=-\log q(y)$. If $p_x=P(Y\mid X=x)$, then the Bayes action is $p_x$
and
\[
    g_x(q)=\KL(p_x\Vert q).
\]
Again the Bayes and regret quotients coincide with posterior equivalence.

\paragraph{Zero--one classification.}
Let $Y\in\{1,\ldots,K\}$ and actions be class labels. Then
\[
    g_x(a)=p_{\max}(x)-P(Y=a\mid x),
    \qquad
    p_{\max}(x)=\max_jP(Y=j\mid x).
\]
The Bayes quotient keeps only the maximizing class, but the complete regret
vector determines the whole posterior. Indeed,
\[
    \sum_a g_x(a)=Kp_{\max}(x)-1,
\]
so
\[
    p_{\max}(x)=\frac{1+\sum_a g_x(a)}{K},
    \qquad
    P(Y=a\mid x)=p_{\max}(x)-g_x(a).
\]
Thus two inputs with the same Bayes label can be distinct consumer-specific source
states even under ordinary classification loss.

\subsection{Relation to the full posterior about the criterion}

Let
\[
    S_Y(X):=P(Y\in\cdot\mid X)
\]
denote the posterior of $Y$ given $X$. The companion TRI-contract paper
\citep{ArmstrongContract2026} uses $S_Y$ for the interface that preserves
everything $X$ tells the receiver about $Y$. Since every action cost depends on
$X$ through this conditional distribution, every action-regret profile is also
a function of $S_Y$.

\begin{proposition}[Posterior, regret and Bayes hierarchy]
\label{prop:posterior-regret-bayes}
There exist deterministic maps $h_{\mathcal C}$ and $k_{\mathcal C}$ such that
\[
    S_Y \longrightarrow G_{\mathcal C}\longrightarrow B_{\mathcal C}
\]
under unique Bayes actions, where the arrows denote deterministic
factorisation.
\end{proposition}

When the criterion and consumer are considered together, the three objects
have a simple ordering. $S_Y$ keeps everything $X$ tells the receiver about $Y$.
$G_{\mathcal C}$ keeps only what can change the relative cost of an action
available to this consumer. $B_{\mathcal C}$ keeps only what is needed to
identify the best action. This ordering connects the two papers. It does not
mean that Task-Sufficient Contraction must always begin by constructing
$S_Y$.

\subsection{Loss decomposition}

Suppose an information interface produces a representation $U$ from $X$, and
a receiving rule produces $A$ from $U$, so $Y-X-U-A$ is a Markov chain. Then
\begin{equation}
    \E[\ell(A,Y)]
    =
    \E[b(X)]
    +
    \E[g_X(A)].
    \label{eq:loss-decomposition}
\end{equation}
The first term is the loss that remains even with full information. The
interface can affect only the second term, which measures the extra loss caused
by choosing a non-optimal action. The next question is whether evaluating that
extra loss really requires the full source $X$.

\section{Task-Sufficient Contraction}

Assume in this section that the action set $\mathcal A$ is finite.
Following the established rate-regret formulation, define
\begin{equation}
    R_X(D)
    :=
    \inf_{\substack{p(a\mid x):\\
                    \E[g_X(A)]\le D}}
    I(X,A).
    \label{eq:raw-rd}
\end{equation}
By data processing, an intermediate representation can be replaced by the
distribution over actions that it induces without increasing mutual
information, as in the single-cycle formulation of
\citet{WalshRateRegret2026}. The source-side result in \citet{ArmstrongIB}
makes the complementary move: it averages an encoder over source states that a
sufficient statistic already treats as identical. The theorem below asks the
same reduce-first question for the consumer-specific source. The consumer is
fixed before the distortion target $D$ is chosen. The question is whether the
same reduced source works for the whole rate-regret curve.

Because $g_x(\cdot)$ is constant inside each $G_{\mathcal C}$ fibre, define
\begin{equation}
    d_G(g,a):=g_x(a)
    \qquad\text{for any }x\text{ with }G_{\mathcal C}(x)=g,
    \label{eq:regret-distortion}
\end{equation}
which is well defined, and let
\begin{equation}
    R_G(D)
    :=
    \inf_{\substack{p(a\mid g):\\
                    \E[d_G(G_{\mathcal C},A)]\le D}}
    I(G_{\mathcal C},A).
    \label{eq:regret-rd}
\end{equation}

\begin{theorem}[Rate-regret invariance under consumer-specific contraction]
\label{thm:contraction}
For every achievable $D$,
\begin{equation}
    \boxed{R_X(D)=R_G(D).}
    \label{eq:contraction}
\end{equation}
\end{theorem}

\begin{proof}
Take any feasible action channel $p(a\mid x)$ and average it within each
regret-equivalence fibre:
\[
    \bar p(a\mid g):=P(A=a\mid G_{\mathcal C}=g).
\]
The averaged channel preserves the joint law of $(G_{\mathcal C},A)$ and hence
preserves expected regret, while data processing gives
\[
    I(G_{\mathcal C},A)\le I(X,A).
\]
Thus $R_G(D)\le R_X(D)$. Conversely, any $G_{\mathcal C}$-channel pulls back to
$X$ by $p(a\mid x):=p(a\mid G_{\mathcal C}(x))$. Then
$A-G_{\mathcal C}-X$ and $G_{\mathcal C}$ is a function of $X$, so
$I(X,A)=I(G_{\mathcal C},A)$ and the expected regret is identical. Hence
$R_X(D)\le R_G(D)$.
\end{proof}

The averaging step is classical sufficiency and rate-distortion machinery
\citep{Bahadur1954,Witsenhausen1980}. What matters here is when it is applied.
The consumer-specific source is fixed before $D$ is selected, yet no point on
the one-step rate-regret curve improves by returning to the richer source.

Under unique Bayes actions, the zero-regret endpoint collapses further to the
Bayes quotient. In the finite setting,
\[
    R_X(0)=H(B_{\mathcal C}),
\]
a result already explicit in the policy-equivalence rate-regret formulation of
\citet{WalshRateRegret2026}. The Bayes quotient therefore describes the zero-regret endpoint. Preserving
the regret of every feasible action can require distinctions absent from
the Bayes quotient. The source $G_{\mathcal C}$ preserves these regrets
pointwise and is sufficient for the entire one-step rate-regret curve.
This does not assert that it is minimal among all statistics that happen
to preserve that optimized curve.

Theorem~\ref{thm:contraction} is stated for finite $\mathcal A$. The
continuous-action quadratic results below do not use this coding theorem.
Instead they compare the full regret functions directly.

\subsection{Five contraction regimes}

In practice, a designer may propose a simpler task variable $Q=q(X)$ before
knowing whether it really can replace the richer source. Five cases should be
kept separate.

\paragraph{Exact consumer-complete contraction.}
If $Q$ is regret-complete, equivalently if
\[
    g_X(a)=d_Q(Q,a)\qquad\forall a,
\]
then Proposition~\ref{prop:coarsest-regret} gives
$G_{\mathcal C}=h(Q)$. The problem can contract at least as far as
$G_{\mathcal C}$. If $Q$ is merely a relabelling or refinement of that quotient,
solving on $Q$ is exact but not minimal.

\paragraph{Architecture-conditioned contraction.}
For an arbitrary deterministic $Q=q(X)$, define the induced distortion
\begin{equation}
    d_Q(q,a):=\E[g_X(a)\mid Q=q].
    \label{eq:induced-distortion}
\end{equation}
If the architecture deliberately restricts the downstream action channel to
factor through $Q$, so $X-Q-A$ is Markov, then
\[
    \E[g_X(A)]=\E[d_Q(Q,A)].
\]
This includes stochastic encoders and intermediate representations generated
from $Q$. The construction is classical indirect distortion
\citep{Witsenhausen1980}. It gives the correct loss for a system that has
chosen to work only through $Q$. It does not show that the consumer itself
would be unable to use information that $Q$ discarded.

\paragraph{Surrogate contraction.}
A proposed $Q$ may instead support a simpler approximate loss
$\widetilde d(Q,a)$ satisfying
\begin{equation}
    \alpha\,\widetilde d(Q,a)
    \le g_X(a)
    \le \beta\,\widetilde d(Q,a)
    \label{eq:surrogate-sandwich}
\end{equation}
in a stated regime. Working on $Q$ then comes with a bound, but it is not an
exact reduction.

\paragraph{Candidate reduction failure.}
A proposed task variable may merge states that still assign different costs to
some admissible action. The variable is then too coarse for exact contraction.

\paragraph{Corrected contraction.}
Failure of the obvious task variable need not mean that the full source is
irreducible. A different reduced source may still preserve the downstream
problem exactly. The constrained GOQ example below is the main instance.

These cases answer different questions. Does the reduced source preserve the
whole later problem being claimed? Is it exact only because the architecture
forces the system to use that source? Does it provide only an approximation?
Or was the wrong reduced source chosen in the first place?

\section{Neighbouring theories and attribution}

Task-Sufficient Contraction sits beside several established ideas rather than
replacing them. Classical sufficiency and Blackwell comparison ask when a
reduced observation preserves decision value
\citep{Blackwell1953,Bahadur1954}. Indirect rate-distortion gives the correct
loss when the encoder observes a variable different from the one on which
performance is judged \citep{Witsenhausen1980}. Functional compression asks
which source distinctions must survive when the receiver needs a function of
the source rather than a reconstruction of the source itself
\citep{OrlitskyRoche2001}.

\citet{Bruhacs2026TaskSufficient} formulates task-relative preservation under
projection and the least quotient sufficient for a declared task family.
Those results provide a close formal comparison for task-relative reduction.
The present paper develops a certification framework for stated downstream
optimization families, together with the consumer-specific regret
construction and affine-action quadratic characterization. The general
idea that a task determines which distinctions must survive is not claimed
as new.

Recent work develops consumer-specific reductions. \citet{Sevetlidis2026}
studies Bayes-sufficient representations for a fixed loss.
Sevetlidis also studies loss shift, in which changing the loss can make a
previously sufficient representation inadequate even when the joint data
distribution remains fixed \citep{Sevetlidis2026LossShift}. This concerns
transfer between declared losses. The rate-regret result here fixes the
consumer and asks whether one source reduction remains exact across the
stated distortion targets.
\citet{RussoVanRoy2022} and \citet{ArumugamVanRoy2021} connect information
acquisition and near-optimality through rate-distortion ideas. Most directly,
\citet{WalshRateRegret2026} uses policy regret as distortion and identifies
policy equivalence as the exact zero-regret endpoint of a single-cycle
rate-regret problem. These results establish important pieces of the
decision-aware picture.

\citet[][Appendix~I]{Hyeon2026} also study Q-sufficiency, which preserves
the optimal Q-function, and derive action sufficiency under a specified
tie-breaking rule. The regret-complete source here preserves relative action
costs while discarding action-independent loss baselines.

\citet{Wei2026WorldModel} distinguish mechanism, response, and decision
sufficiency for world-model planning. For a fixed query and candidate set,
their response object retains absolute candidate costs, while the decision
object retains a selected optimum. With the candidate set and cost
criterion matched, the absolute cost profile determines the regret profile
used here, which in turn determines the optimum under this paper's
unique-optimum assumption. Their planner analysis also identifies
normalized relative-cost weights as sufficient for MPPI updates and
explicitly notes common-offset invariance. Thus relative-cost preservation
is not absent from their analysis. Their work distinguishes the information
required by particular planner operations, including search and final
selection. Theorem~\ref{thm:contraction} here
certifies preservation of the complete one-step rate-regret curve by one
fixed regret-complete source.

The common question developed here is whether a task-certified reduction
can be fixed before the later rate, codebook, distortion target, or other
operating choice, while preserving the stated downstream problem family.
The framework distinguishes exact consumer-complete contraction,
architecture-conditioned contraction, bounded approximation, candidate
failure, and corrected contraction. The constructions in this paper show
how those different guarantees can be established and why they should not
be conflated.

Related equivalence ideas appear elsewhere. Myhill--Nerode minimisation merges
automaton histories that no permitted continuation can distinguish with
respect to language acceptance \citep{Nerode1958}. Observational equivalence
separates internal identity from external behaviour in process theory
\citep{HennessyMilner1985}. Bisimulation-based state abstraction preserves
controlled behaviour over transitions \citep{Ferns2004}. These are useful
analogies. Their dynamic closure requirements are not part of the static
theory developed here.

\section{Three earlier reductions}

The three reductions in this section are earlier results, not three new
contributions of the present paper. Their comparison exposed the common
reduce-first pattern. Each case is therefore presented in two parts: what the
earlier paper already established, and what the comparison adds here.

\subsection{Information Bottleneck: exact pre-optimization contraction}

\paragraph{Native result.}
Let \(T\) be the source and \(C\) the relevance variable. Consider the
probabilistic prediction consumer with action space \(\Delta(\mathcal C)\) and
logarithmic loss. For
\[
    p_t=P(C\mid T=t),
\]
the Bayes action is \(p_t\), and the regret of predicting \(q\) is
\[
    g_t(q)=\KL(p_t\Vert q).
\]
Under log loss, two source states are therefore equivalent exactly when they
induce the same posterior over $C$.

The existing source-side sufficiency theorem assumes a deterministic
\(\phi(T)\) satisfying
\[
    C\perp T\mid\phi(T).
    \label{eq:ib-suff}
\]
Thus the conditional distribution of $C$ is determined by $\phi(T)$, so
$\phi(T)$ contains everything needed to compute the log-loss regret of any
prediction.

The companion IB result proves that every full encoder has a fibre-averaged
shadow on \(\phi(T)\) with the same relevance and no greater source rate, and
that
\begin{equation}
    I(U,T)
    =
    I(U,\phi(T))
    +
    I(U,T\mid\phi(T)).
    \label{eq:ib-decomp}
\end{equation}
The complete IB curve, the Lagrangian infima, the minimisers and the
finite-block log-loss problem are therefore unchanged by solving first on
\(\phi(T)\) \citep{ArmstrongIB}.

In the present terminology, IB is the clean exact case:
\[
    \boxed{
    \text{criterion-certified quotient}
    \longrightarrow
    \text{exact optimization contraction}.
    }
\]

\paragraph{TSC reading.}
The native theorem says more than that a good compressed representation
exists. It says that the IB optimization never needed all distinctions in
$T$ in the first place. This makes IB the cleanest exact example of contraction before a later operating point is chosen. The theorem itself remains the result of the earlier
paper.

\subsection{Semantic rate-distortion: architecture-conditioned contraction}

\paragraph{Native result.}
The semantic/source-coding example starts from a producer state \(T\) and a
designer-chosen deterministic oracle allocation
\[
    Q=\phi^*(T).
\]
The underlying consumer is evaluated by welfare loss \(\ell(t,a)\).

Under general smooth concave utility, knowing the deterministic oracle
$\phi^*(T)$ is not enough by itself to recover the cost of every other action.
The semantic rate-distortion paper therefore defines the reduced loss as
\begin{equation}
    d_s^*(q,a)
    :=
    \E[\ell(T,a)\mid\phi^*(T)=q].
    \label{eq:sk-induced}
\end{equation}
Any reconstruction rule on the full source can then be averaged over states
with the same $Q=\phi^*(T)$. This leaves the reduced loss unchanged and cannot
increase mutual information. Conversely, any rule defined on $Q$ can be used
on $T$ through the map $T\mapsto Q$. The semantic marginal rate-distortion
problem can therefore be solved exactly on $\phi^*(T)$ with the induced loss
$d_s^*$ \citep{Witsenhausen1980,ArmstrongSK}.

Here $d_s^*$ is the expected task loss after conditioning on the reduced
source, matching the semantic rate-distortion formulation. Written in the
regret notation used here, the corresponding loss is
\[
    d_Q^{\rm reg}(q,a)
    =d_s^*(q,a)-\E[b(T)\mid \phi^*(T)=q].
\]
The subtracted term does not depend on the chosen action once $q$ is fixed.
It therefore changes the baseline loss but not which action is better. The
raw-loss and regret versions describe the same reduced decision problem.

For $\phi^*(T)$ itself to preserve the cost of every possible mistake,
additional structure is needed. It holds, for example, when
\[
    \ell(t,a)-\ell(t,\phi^*(t))
\]
is a function only of \((\phi^*(t),a)\), including the exact quadratic
subclasses considered in the companion work.

The same example exposes a second issue. One interface may have to support
both allocation and verification. If verification can be expressed as an
ordinary per-item distortion, the two requirements still fit the same
semantic rate-distortion problem. If verification depends on the occupancy of
the whole pool, it no longer has that form, and the problem becomes a
constrained feasibility problem instead \citep{ArmstrongSK}.

Thus the semantic example separates two statements:
\[
    \boxed{\text{a quotient can define an exact restricted problem}}
\]
without implying
\[
    \boxed{\text{the chosen variable is the smallest exact source for the
    unrestricted consumer}.}
\]

\paragraph{What the comparison adds.}
Comparing this case with IB exposes an important distinction. Information can
be absent because the consumer has no use for it, or because the system design
has chosen not to expose it. Those are not the same claim. The additional
verification requirement shows the same point from another direction: a
reduction that is sufficient for one declared requirement need not be
sufficient for another.

\subsection{Goal-oriented quantization: failure of the obvious source and exact correction}

\paragraph{Native result.}
The current GOQ manuscript, \emph{Task-Oriented Quantization for Quadratic
Scheduling: Centroid Water-Filling and Power-Diagram Encoders}, supplies the
constrained scheduling case \citep{ArmstrongGOQ}. Let $T$ be the producer
state and $\phi^*(T)$ the
deterministic full-information optimum.

\paragraph{Exact interior case.}
For isotropic quadratic task loss,
\[
    \ell(t,a)=\frac{\gamma}{2}\|a-\phi^*(t)\|^2,
\]
the entire regret profile factors through $\phi^*(t)$. The Bayes and regret
quotients therefore coincide with the oracle action, and the task RDF and
finite-$K$ quantization problem reduce exactly to ordinary MSE coding on
$\phi^*(T)$.

\paragraph{Smooth interior case.}
Under a zero-gradient interior optimum and smooth/strong curvature,
\[
    \frac{\alpha}{2}\|\phi^*(t)-a\|^2
    \le
    \ell(t,a)-\ell(t,\phi^*(t))
    \le
    \frac{\beta}{2}\|\phi^*(t)-a\|^2.
\]
Distance between oracle actions then gives a controlled approximation to the
true task loss. It yields the corresponding rate-distortion bounds and
approximation guarantee, but it does not make the oracle action an exact
reduced source.

\paragraph{What the comparison adds.}
The exact and smooth interior results make the oracle action look like a
natural reduced source. The constrained case shows why choosing that source is
a real problem rather than a change of terminology. Two source states can have
the same optimal action while making other feasible actions cost different
amounts. The answer is not to keep the whole raw source. It is to find the
smaller source that still preserves those differences.

\paragraph{Quadratic losses on affine action sets.}
The correction yields a more general source rule. Let the feasible
action set $\mathcal A\subseteq\mathbb R^n$ have affine hull
\[
    \operatorname{aff}(\mathcal A)=a_0+V,
\]
where $V$ is the linear subspace spanned by feasible action differences, and
assume $\mathcal A$ is nonempty and that $\inf_{b\in\mathcal A}f(t,b)$
is finite for every source state under consideration, with
\[
    f(t,a)=\|t+a\|^2,
    \qquad
    g_t(a)=f(t,a)-\inf_{b\in\mathcal A}f(t,b).
\]

\begin{theorem}[Affine-action quadratic regret quotient]
\label{thm:affine-quadratic}
For any two source states $t,t'$,
\[
    g_t(\cdot)=g_{t'}(\cdot)
    \quad\Longleftrightarrow\quad
    t'-t\in V^\perp.
\]
Consequently the orthogonal projection
\[
    \boxed{\Pi_V t}
\]
is a canonical representative of the regret-profile quotient.
\end{theorem}

\begin{proof}
Let $d=t'-t$. For any feasible $a$,
\[
    f(t',a)-f(t,a)
    =\|t'\|^2-\|t\|^2+2d^\top a.
\]
If $d\in V^\perp$, then $d^\top a$ is constant on the affine set
$a_0+V$, so every feasible action value shifts by the same amount and the
regret profile is unchanged. Conversely, if the regret profiles agree, then
$f(t',a)-f(t,a)$ is constant over $a\in\mathcal A$. Hence
$d^\top(a-a')=0$ for every feasible $a,a'$. These differences span $V$, so
$d\in V^\perp$.
\end{proof}

The theorem has a simple interpretation. A component of the source that is
orthogonal to every direction in which feasible actions can differ may change
the absolute objective value, but it changes every feasible action by the same
amount. It cannot change the cost of one action relative to another, so the
consumer does not need it.

\paragraph{Fixed-energy scheduling.}
Fix a positive energy budget $E>0$.
For load states $t\in\mathbb R_{\ge0}^N$ and
\[
    \mathcal A_E=\{a\in\mathbb R_{\ge0}^N:\mathbf1^\top a=E\},
\]
the affine direction is $V=\mathbf1^\perp$. Theorem~\ref{thm:affine-quadratic}
therefore yields
\begin{equation}
    \boxed{
    G_{\mathcal C}(t)
    =Pt
    =t-\bar t\mathbf1,
    \qquad
    P=I-\frac1N\mathbf1\mathbf1^\top.}
    \label{eq:centered-load}
\end{equation}
This is centered load.

The Bayes-optimal action is the water-filling solution
\[
    \phi^*(t)=\operatorname{WF}_E(t)=(\lambda(t)\mathbf1-t)^+,
\]
but the oracle action is strictly too coarse once non-optimal actions are
possible.
For example, with $E=1$,
\[
    t=(0,10),\qquad t'=(0,20),
\]
both have $\phi^*(t)=\phi^*(t')=(1,0)$, yet for the feasible action $a=(0,1)$,
\[
    g_t(a)=20,\qquad g_{t'}(a)=40.
\]
Thus identical Bayes actions do not imply identical regret profiles.

At the same time, the raw load contains a common shift that cannot change
the relative cost of feasible actions. The exact GOQ assignment rule for fixed
representatives therefore depends only on centered load because
\[
    2t^\top a_k+\|a_k\|^2
    =
    2(Pt)^\top a_k+\|a_k\|^2+2\bar tE,
\]
and the last term is common to every $k$. The cell update can also be
computed from centered load alone. If
\[
    \mu_k=\E[T\mid T\in\mathcal C_k],
\]
then
\[
    \mu_k=\E[PT\mid\mathcal C_k]
          +\E[\bar T\mid\mathcal C_k]\mathbf1,
\]
and water-filling is invariant to common shifts, so
\[
    \operatorname{WF}_E(\mu_k)
    =
    \operatorname{WF}_E\!\left(\E[PT\mid\mathcal C_k]\right).
\]
Hence both assignment and update are well defined on $PT$.

On the natural load state space, the constrained example therefore has the form
\[
    \boxed{
    T\text{ contains a null common mode}
    \quad\longrightarrow\quad
    T-\bar T\mathbf1\text{ is exact}
    \quad\longrightarrow\quad
    \phi^*(T)\text{ can be too coarse}.}
\]
Failure of the obvious task representation does not imply that the producer
source is irreducible. It can mean that the task induces a different source.

\begin{table}[!hb]
\centering
\small
\begin{tabularx}{\textwidth}{p{0.16\textwidth}XXX}
\toprule
 & IB & Semantic/RD & GOQ\\
\midrule
Producer source
& $T$
& demand/state $T$
& task/load state $T$\\

Declared consumer
& predictor of $C$
& allocation/semantic consumer
& action consumer\\

Reduction status
& exact reduction
& architecture-conditioned generally
& exact, approximate, or corrected reduction\\

Key object
& sufficient $\phi(T)$
& designed $\phi^*(T)$ with induced loss
& oracle action or projected load\\

Boundary
& insufficient statistic loses relevance
& second contract / architecture restriction
& active constraints change the right reduced source\\
\bottomrule
\end{tabularx}
\caption{Three separately developed reductions viewed through the
Task-Sufficient Contraction framework.}
\label{tab:audit}
\end{table}

\section{Machine information interfaces}

The same idea has a direct systems interpretation. A producer may keep a rich
internal representation $H_P$, while the receiving machine uses a different
internal representation $H_C$. An interface between them does not necessarily
need to expose either internal representation.

One approach is representation alignment. A sender and receiver may exchange
embeddings, hidden states, or caches and learn a transformation that makes
the sender's internal state usable by the receiver. Recent latent-communication
work makes this approach increasingly explicit, including cross-architecture
alignment as an open design problem \citep{LiuLatent2026}. This can be
valuable when the receiving system needs broad access to the sender's latent
content.

Task-Sufficient Contraction suggests another possibility. If the receiver has a declared task, an external interface can be built around a certified reduced source rather than either party's complete internal representation. Two producers with different internals may then be able to serve the same receiver without first making those internals mutually intelligible.

This distinction is particularly natural for agent protocols in which the
remote system is intentionally opaque. A2A exposes capabilities and exchanged
artifacts while allowing internal state, memory, tools, and proprietary logic
to remain hidden \citep{A2A2026}. In that setting the practical question is
not how to reproduce the remote agent internally. It is what information must
cross the interface for the receiving task.

There need not be one universal interface for every receiver. A broad task
may still require a rich representation, while a narrow task may require much
less. One way to describe the alternative is \emph{consequence alignment}:
producers need not represent the world in the same way if they expose the
distinctions that can change what the receiver should do.

This does not make internal representations irrelevant, and it does not imply
that the reduced interface is cheap to compute. A rich latent representation
may still be the best internal way to produce the information the receiver
needs.

The systems implication is narrower:

\begin{quote}
Representation alignment can be sufficient for interoperability, but it need
not be necessary when the receiving task admits a shared task-sufficient
interface.
\end{quote}

That proposition is testable. A later systems study can compare a
representation-alignment pipeline with an interface constructed directly from
a verified reduced source, measuring task performance, transmitted bits
or tokens, latency, compute, memory and energy. The present paper stops before
making those empirical claims.

\section{Scope and boundaries}

Every reduction in this paper has a stated scope. Exactness never means that the discarded information is useless for every possible future task.

If one interface must serve consumers
\[
    \mathcal C_1,\ldots,\mathcal C_m,
\]
with regret sources
\[
    G_1,\ldots,G_m,
\]
then the joint source
\[
    G_{1:m}:=(G_1,\ldots,G_m)
\]
preserves the regret of every available action for all of them.

\begin{proposition}[Static multi-consumer refinement]
Adding a declared consumer can only refine the joint regret source:
\[
    G_{1:m}
    =
    h(G_{1:m+1})
\]
for the projection \(h\). In finite systems,
\[
    H(G_{1:m})
    \le
    H(G_{1:m+1}).
\]
\end{proposition}

This is a static result. An interface that must serve several declared tasks
may need to keep distinctions that any one of those tasks could ignore. The
present paper does not claim that such an interface can anticipate tasks that
were never part of its specification.

Several other boundaries matter.

\paragraph{Hidden or misspecified criteria.}
The reduction is only as good as the consumer specification. If the declared
actions or loss leave out a consequence that actually matters, the reduced
source may discard information the real problem needs.

\paragraph{Architecture versus consumer.}
A designer may intentionally force a system to work only through a reduced
variable \(Q\), even when the consumer could use more information. The loss in
\eqref{eq:induced-distortion} then describes that restricted system correctly.
It does not show that the consumer itself had no use for the discarded
information.

\paragraph{Approximate reduction.}
A simpler reduced variable may give useful bounds without preserving exact
regret. Smooth GOQ is an example. A bounded approximation should not be
described as exact sufficiency.

\paragraph{Implementation cost.}
Information-theoretic redundancy is not the same as implementation cost.
Removing a distinction from the external task does not prove that the cheapest
implementation should avoid computing or storing it internally.

\paragraph{State-changing action and the recursive boundary.}
The present theory is static. Bisimula\-tion-based state abstraction requires
equivalence to survive state transitions \citep{Ferns2004}.
\citet{Subramanian2022AIS} separately require sufficiency for expected-reward
evaluation and for predicting the compressed state's own evolution in their
information-state formulation. \citet{WalshConsequence2026} studies recurrent
consequence-sensitive compression, including policies that alter later
inferential conditions. The present paper does not impose or solve these
recursive requirements. If acting through an interface changes the learner,
state distribution, criterion, candidate set, action menu, or dependence
structure, a reduction valid for the current problem may no longer be valid
for the next decision. The results here do not compare source requirements
across such changes or define a dynamic update rule for the task-sufficient
source.

\section{Discussion}

The weather example makes sufficiency concrete. A forecast that cannot settle
whether a gust will dislodge a roof tile can still provide everything needed
for a particular wind-turbine decision. Sufficiency belongs to the declared
use of the information, not to the richness of the thing being described.

This becomes especially important when the consumer is another machine. An
agentic economy creates repeated handoffs among models, agents, controllers,
and services that may be opaque to one another. If a receiving machine can be
given a certified task-sufficient interface, interoperability need not require
sharing or aligning complete internal representations at every handoff. The
systems value of a valid reduction therefore grows as machine-to-machine
coordination becomes more common and heterogeneous.

Task-Sufficient Contraction says that the reduction can sometimes be made
before the later optimization begins. First identify the source distinctions
the task can never use. Then solve the later information problem on what
remains. An exact result means that returning to the richer source cannot
improve any of the operating points covered by the theorem.

The three earlier reductions show why the framework needs more than one label.
The Information Bottleneck provides the cleanest exact case. A sufficient
source statistic can be fixed before the rate or Lagrange multiplier is
selected, and the complete optimization family is unchanged. Designed-source
semantic rate-distortion shows that an exact reduced problem can instead be
created by an architectural restriction and an induced loss. Goal-oriented
quantization shows that an obvious task representation can fail without making
the raw source irreducible. Under fixed-energy quadratic scheduling, centered
load is exact even though the optimal water-filled action is too coarse.

The consumer-specific regret construction gives another exact case. It keeps
every distinction that can change the relative cost of an available action and
discards distinctions that cannot. If the same problem is also described by a
criterion $Y$, the relationship is:
\[
    X\longrightarrow S_Y\longrightarrow G_{\mathcal C}
    \longrightarrow B_{\mathcal C}.
\]
The chain separates three questions. What information does $X$ provide about $Y$? Which
of those distinctions can change the cost of an action available to this
consumer? Which distinctions are still needed if only the best
full-information action matters?

For machine interfaces, the implication is that shared internal
representations are not always necessary. A producer may keep a rich internal
state while exposing only a task-sufficient external object. Another producer
with a different internal representation may expose the same external
distinctions. This does not imply that every task admits a small interface or
that a smaller interface is cheaper to compute. It says only that the external
information requirement should be derived from the receiving problem rather
than assumed to be the producer's full representation.

The static boundary matters. A reduction that is valid now need not remain
valid after the receiver acts. Action may change the learner, state
distribution, candidate pool, action menu, criterion, or dependence structure.
The next question is therefore not another static contraction at one fixed
state. It is how the information required by the task changes across states,
and whether the representation should become finer or coarser as the
action-consequence problem changes. This paper does not answer that recursive
question.

\bibliographystyle{plainnat}
\bibliography{task_sufficient_contraction_v2_10_final_fixes_2026-10-01}

\end{document}